\documentclass[12pt, center]{article}
\usepackage{float}
\usepackage{amsfonts}
\usepackage{mathrsfs}
\usepackage{amsmath}
\usepackage{amsthm}
\usepackage{amssymb}
\usepackage{booktabs}
\usepackage{caption}
\usepackage{multirow}
\usepackage{cases}
\usepackage{mathrsfs}
\usepackage{graphicx}
\usepackage{array}
\usepackage{color}
\usepackage{needspace}
\usepackage{rotating}
\usepackage{threeparttable}
\usepackage{hyperref}
\usepackage[utf8]{inputenc}

\newtheorem{thm}{\bf Theorem}

\newtheorem{remark}{\bf Remark}
\newtheorem{exam}{\bf Example}
\usepackage{mathtools}
\usepackage{bm}
\usepackage{calc}
\usepackage{longtable}
\usepackage{setspace}
\usepackage[authoryear,round]{natbib} 

\begin{document}
 \oddsidemargin = -5pt \headsep=-10pt \footskip=50pt

\title{\bf\Large Estimation and Hypothesis Testing of Fixed Effects Models-Based Uncertainty for Factor Designs \\
}

\author{\normalsize Fan Zhang,  Zhiming Li\footnote{Corresponding author: zmli@xju.edu.cn}\\
{\em { \small College of Mathematics and System Science,  Xinjiang University,   Urumqi 830046,  China}}}
\date{}
\maketitle

\noindent
{\bf Abstract:} 
To analyze the uncertain data frequently encountered in practice, this paper proposes novel fixed-effects models that incorporate an uncertain measure to investigate variables of interest and nuisance variables in factor designs. First, an uncertain fixed-effects (UFE) model of a single-factor design is established, and uncertain estimation and hypothesis testing are conducted. We then extend the UFE model to two-factor designs with and without interactions and classify them as balanced or unbalanced based on the equality of replicates within each combination. In the above UFE models, the effectiveness and practicality of estimation and hypothesis methods are demonstrated through three real-world cases, including both balanced and unbalanced designs. These examples highlight the models' ability to handle uncertain experimental data.

\begin{flushleft}
{\bf Keywords:}   Factor design, fixed effects, uncertain measure, uncertain model. 
\end{flushleft}

\noindent \hrulefill
\section{Introduction}
The design of experiments is widely applied in agriculture (\citet{TorresMercado2024}), medicine (\citet{Muguruma2022}),  engineering (\citet{Liew2022}), manufacturing (\citet{Rahmatabadi2023}), and scientific research (\citet{Alexanderian2021}). An experiment often involves altering one or more factors that may affect the outcome under controlled conditions, and these designs are referred to as single-factor or multi-factor designs. Its core lies in systematically arranging experiments to investigate the relationships between the research subject and influencing factors, thereby providing a reliable foundation for data analysis. To analyze significant differences among factor effects, fixed-effects, random-effects, or mixed-effects models are commonly used in multi-factor designs. The fixed-effects (FE) model is a panel-data regression model.  Unlike the FE model, the random effects model treats individual-specific heterogeneity as part of the error terms,  pioneered by \citet{Eisenhart1947}. The mixed-effects model combines the first two types (\citet{Laird1982}). Among these models, fixed effects are regarded as a ubiquitous and powerful tool for eliminating unwanted variation (\citet{BD2024}). 

 Following \citet{Dean2017}, we review the classical FE model within the framework of probability theory. Take a single-factor design as an example. Let $y_{ij}$ be the $j$th $(j=1,\ldots, n)$ observation taken under the $i$th $(i=1,\ldots, a)$ level or treatment of a factor $A$. Let $\mu$ be the overall mean of all treatments, and $\mu_i(i=1,\ldots, a)$ be the mean of the $i$th level of the factor $A$.  Under the assumptions of independent and identically distributed observations, and normality, no corrections, and homoskedasticity of random errors,  $y_{ij}=\mu_i+\varepsilon_{ij}(i=1,\ldots, a, j=1,\ldots, n)$ is called a mean model, where $\varepsilon_{ij}$ is a random error and $\varepsilon_{ij}\sim N(0,\sigma^2)$. Denote $a_i=\mu-\mu_i(i=1,\ldots, a)$, a parameter unique to the $i$th level called the $i$th level effect. Thus, an alternative model, called the FE model, is established by
\begin{equation*}
\label{OneFE}
y_{ij}=\mu+a_i+\varepsilon_{ij}, i=1,\ldots, a, j=1,\ldots, n, \sum_{i=1}^aa_i=0,
\end{equation*}
which is also called the one-way or single-factor analysis of variance model. 
This framework was formalized by \cite{Scheffe1959} and extended to multi-factor designs, in which analysis of variance (ANOVA) tests the significance of factor effects. Subsequent research has addressed key complexities: handling unbalanced data via methods such as exact permutation tests and generalized linear models (\cite{KheradPajouh2010, Thiel2017}); addressing high-dimensional problems via regularization techniques (\cite{Engel2020}); managing heteroscedasticity with robust methods like the variance-weighted $F$-approximation and its refinements (\cite{Welch1951, Alexander1994, Pilkington2024}); employing developed multiple comparison procedures (\cite{Tukey1949, Hsu1996, Mee2023}); and utilizing residual analysis techniques for model diagnostics (\cite{Anscombe1963, Montgomery2017}).
Numerous theoretical achievements and applications of FE models can be referenced in \citet{Wilk1955}, \citet{Pietraszek2016}, 
and \citet{Passerine2024}. 
However, some challenges arise in experimental observations: small sample sizes by design and frequent inherent uncertainty, rather than genuine randomness, render probability-based statistical tools less applicable. Empirical evidence suggests that real-world data often fail to satisfy core assumptions in probability theory (including normality, frequency stability, independence, and randomness). When the sample size is limited, the data are imprecise, or uncertainty originates from human belief and cognition, probability theory may no longer be adequate in real-world problems, as noted in numerous studies, such as carbon spot price (\citet{LiuLi2024}), $\text{O}_{3}$ concentration (\citet{XiaLi2025}),  epidemic spreading (\citet{XieLio2024}),  and population (\citet{YangLiu2024}). 

The existence of these problems necessitates novel measures to provide a more suitable modeling framework. To solve the challenges mentioned above, \citet{Liu2007} first proposed an uncertain measure grounded in the axioms of normality, duality, subadditivity, and product, further establishing uncertainty theory. As a notable practical extension of this theory, \citet{Liu2010} pioneered the use of uncertain statistics to address the collection, analysis, and interpretation of uncertain data in real-world applications. Currently, uncertain regression models are popular in the field of uncertain statistics. Uncertain statistical inference for these models encompasses two core tasks: parameter estimation and hypothesis testing.  For uncertain regression models, the estimation methods include least squares (\citet{YaoLiu2018}), moment (\citet{LioLiu2018}),  uncertain maximum likelihood (\citet{LioLiu2020}), least absolute deviations estimation (\citet{LiuYang2019}), and ridge estimation (\citet{ZhangGao2024b}). For some complex regression models, local polynomial and B-spline estimation, as well as two-step estimation, are employed (\citet{DingZhang2021}, \citet{ZhangLi2025a}). On the other hand, hypothesis testing under uncertainty is employed to validate the results. Unlike probabilistic methods that rely on conventional statistics, this approach directly utilizes raw data. \citet{YeLiu2022} introduced the framework for testing regression coefficients in uncertain models, demonstrating its superiority in handling non-probabilistic residuals. Further advances by \citet{YeLiu2023b} extended the method to multivariate settings, demonstrating robust performance when stochastic tests fail. Moreover, \citep{ZhangLi2025b} conducted uncertainty-based homogeneity and common tests for multiple finite populations, thereby broadening the applicability of inference under uncertainty. 

Based on the above analysis, numerous valuable research contributions have been made to the field of uncertain regression models. Uncertain regression models based on experimental design often present a distinct yet closely related scenario---it involves
different treatment combinations, balanced/unbalanced structures, and prespecified
potential interaction effects. 
Nevertheless, the underlying uncertainty in such designs has received little attention, and existing uncertain regression toolboxes do not accommodate these features.
To bridge this gap, this paper systematically introduces uncertainty theory into the FE framework of experimental design. To validate effectiveness,
we reanalyse ``presumed random'' data with uncertain models. The main
innovations are summarised as follows:

(i) A system of uncertain fixed effects (UFE) models is first proposed for both single- and two-factor designs, thereby providing a dedicated theoretical framework for analyzing uncertain experimental data.

(ii) Parameter estimation methods, integrating generalized linear models with uncertain vectors, are proposed for unbalanced designs, which overcame the limitation of traditional least squares and achieved explicit estimation and derived distributions.

(iii) The homogeneity test of effects in UFE models is conducted by adjusting data for fixed effects to maintain inter-group independence, thereby 
ensuring reliable results.

(iv) Through three diverse real examples spanning different designs, we substantiate the suitability and practical utility of the proposed methodologies.
\begin{table}[ht]
\centering\small
\renewcommand{\arraystretch}{1}
\caption{Comparison of probabilistic and uncertain fixed effects (FE) models.}
\label{tab:comparison}
\begin{tabular}{>{\centering\arraybackslash}m{2.6cm}>{\centering\arraybackslash}m{5.95cm}>{\centering\arraybackslash}m{5.95cm}}
\hline
\textbf{Aspect} & \textbf{Probabilistic  FE Model} & \textbf{Ucertain FE Model} \\
\hline
Measure & Probability measure  & Uncertain measure  \\
Error & Normal distribution (thin-tailed) & Normal uncertainty distribution (heavy-tailed) \\
Hypotheses & Test statistics ($F/t$ test) & Raw data, inverse uncertainty distribution \\
Sensitivity in\linebreak small samples & Low (low statistical power) & High (stricter rejection criterion) \\
Applicable\linebreak scenarios & Large samples, ideal assumptions satisfied & Small samples, violated assumptions, subjective uncertainty \\
\hline
\end{tabular}
\end{table}

As summarized in Table \ref{tab:comparison}, the UFE model advances beyond traditional probabilistic FE models by fundamentally not relying on frequency stability or randomness assumptions. This shift confers several key advantages: (a) it offers enhanced tolerance to extreme values through its heavy-tailed normal uncertainty distribution (\citet{GuoGuoThiart2010}), making it better suited for small-sample scenarios; (b) it employs a stricter rejection criterion based on direct evaluation of raw data against the null hypothesis distribution, providing higher sensitivity to detect systematic differences often missed by traditional methods; and (c) its applicability in scenarios with subjective uncertainty or unverifiable assumptions complements traditional approaches, particularly when data are limited.

The remainder of this paper is organized as follows.  Section \ref{sec3} establishes an uncertain fixed-effects model for a single-factor design, investigates parameter estimation, and explores the testing process for related hypotheses. The parameter estimation and uncertain hypotheses across two factors are analyzed in Section \ref{sec4}, accounting for interaction effect and replication variability. Section  \ref{sec5} demonstrates the practical effectiveness of the proposed methods through three real examples, and Section  \ref{sec6} provides a brief conclusion.
\section{Uncertain fixed-effects model with single factor}\label{sec3}

In this section, we first review the basic concept of an uncertain measure  (\citet{Liu2023b}), and then propose a fixed-effects model with uncertainty for single-factor designs. 
\subsection{Basic concepts}
Let $\Gamma$ is a nonempty set,  and $\mathcal{L}$ be a $\sigma$-algebra over $\Gamma$.  A map $\mathscr{M}:  \mathcal{L}\mapsto [0,1]$ is called an uncertain measure, satisfying four axioms:

(i) (Normality axiom) $\mathscr{M}\{\Gamma\} = 1$ for the universal set $\Gamma$.

 (ii) (Duality  axiom)  $\mathscr{M}\{\Lambda\} +\mathscr{M}\{\Lambda^c\} = 1$ for any event $\Lambda$. 

 (iii) (Subadditivity  axiom) For every countable sequence of events $\Lambda_1$,$\Lambda_2$,..., we have
  \begin{equation*}
 			\mathscr{M}\left\{\bigcup\limits_{i=1}^\infty \Lambda_i\right\} \leq \sum\limits_{i=1}^\infty \mathscr{M}\{\Lambda_i\}.
 		\end{equation*} 

(iv) (Product  axiom) Let ($ \Gamma_i$, $\mathcal{L}_i$, $\mathscr{M}_i$) be a sequence of uncertainty spaces and $\Lambda_i$ are arbitrarily chosen events from $\mathcal{L}_i$ for $i$ = 1, 2,..., then the product uncertain measure $\mathscr{M}$ is an uncertain measure satisfying
 	$
 			\mathscr{M}\{\prod\limits_{i=1}^\infty \Lambda_i\} = \bigwedge\limits_{i=1}^\infty \mathscr{M}\{\Lambda_i\}.
	$

 For any Borel set $B \subseteq \mathbb{R}$,  $\xi$ is called an uncertain variable if the set $\{\gamma \in \Gamma \mid \xi(\gamma) \in B\} \in \mathcal{L}$.  The function $\Phi(z; {\bm\theta}) = \mathscr{M}\{\xi \leq z\}(z\in \mathbb{R})$ is called an uncertainty distribution with paramater ${\bm\theta}(\in \Theta)$ of $\xi$. Under regular conditions, the inverse function of $\Phi(z;{\bm\theta})$ is called  an inverse uncertainty distribution, denoted by $\Phi^{-1}(\alpha; {\bm\theta})$ for $\alpha \in (0,1)$.  An uncertain variable $\xi$ is  normal if its uncertainty distribution and inverse function are defined as:
 \vspace{0cm}
 \begin{align*}
 \Phi(z;\mu,\sigma) &= \left(1 + \exp\left(\frac{\pi(\mu-z)}{\sqrt{3}\sigma}\right)\right)^{-1}, \quad z \in \mathbb{R},\\
 \Phi^{-1}(\alpha;\mu,\sigma) &= \mu + \frac{\sqrt{3}\sigma}{\pi}\ln\left(\frac{\alpha}{1-\alpha}\right), \quad \alpha \in (0,1),
 \end{align*}
 where $e$ and $\sigma$ are real numbers with $\sigma > 0$, denoted by  $\xi \sim \mathscr{N}(e, \sigma)$. When $e=0$ and $\sigma=1$, $\xi$ is a standard normal uncertainty distribution.	 A regular uncertainty distribution family $\{\Phi(z; \theta) : \theta \in \Theta\}$ is said to be nonembedded for $\theta_0 \in \Theta$ at level $\alpha$ if
\[
\Phi^{-1}(\beta; \theta_0) > \Phi^{-1}(\beta; \theta) \quad \text{or} \quad \Phi^{-1}(1-\beta; \theta) > \Phi^{-1}(1-\beta; \theta_0), \mbox{ some } \theta \in \Theta, 0 < \beta \leq \alpha/2.
\]
The normal uncertainty distribution family $\{\mathscr{N}(e, \sigma) : e \in \mathbb{R}, \sigma > 0\}$ is nonembedded for any $(e_0, \sigma_0)$ at a significance level $\alpha$ (\cite{YeLiu2022}).

\subsection{Uncertain fixed-effects model}
 In an experiment, suppose there is only one factor $A$ with $r$ levels $A_1, A_2, \ldots, A_r$.  Let $m_i$ be the repeated times at level $A_i$ for $i = 1, 2, \ldots, r$, with a total of $N=\sum_{i=1}^{r} m_i$.
If $m_1 = m_2 = \cdots = m_r = m$, it is called a balanced design. Otherwise, it is unbalanced. Let $z_{ij}(i=1,\ldots,r,j=1,\ldots, m_i)$ be observations of the \( j \)th replicate at the \( i \)th level. Based on this, an uncertain model of a single-factor design is defined as follows
\begin{eqnarray}
\label{Mmodel}
  z_{ij} = \mu_i + \varepsilon_{ij}, \quad  i = 1, 2, \ldots, r; \ j = 1, 2, \ldots, m_i,  
\end{eqnarray}
where $\mu_i$ is the mean of $A_i$, and $\varepsilon_{ij}$ are  mutually independent and identically distributed uncertain errors for the $j$th experiment at the $i$th level.

\begin{thm}\label{thm1}
 In the model (\ref{Mmodel}), let $\hat{\mu}_i(i=1,...,r)$ be the least-squares estimator of $\mu_i$. Then, $\hat{\mu}_i = \frac{1}{m_i} \sum_{j=1}^{m_i} z_{ij}$.    
\end{thm}
\begin{proof}
By employing the least-squares method (\cite{YaoLiu2018}), $\hat{\mu}_i(i=1,...,r)$ is obtained by solving
\[
\bm{\hat\mu} = \mbox{argmin}_{\bm\mu}\sum_{i=1}^r \sum_{j=1}^{m_i} \left( z_{ij} - \mu_i \right)^2,
\]
where \(\bm{\hat\mu}=(\hat\mu_1,\hat\mu_2,...,\hat\mu_r)\).
\end{proof}

To perform subsequent confidence interval construction 
and hypothesis testing, we assume $\varepsilon_{ij} \sim \mathscr{N}(0, \sigma_0)$, 
where $\sigma_0$ is determined and validated through the residual analysis 
procedure described in Remark~\ref{rem1}.

\begin{remark} \label{rem1} { After obtaining the point estimates, we verify the distribution of all residuals $\hat{\varepsilon}_{ij} = z_{ij} - \hat{\mu}_i$ in model (\ref{Mmodel}) to justify ${\varepsilon}_{ij} \sim \mathscr{N}(0, \sigma_{0})$.
First, it is confirmed that  $\hat{\varepsilon}_{ij}\sim \mathscr{N}(0, \sigma_{i0})$ using Corollary~1 in \cite{YeLiu2022}.
Then, following \cite{ZhangLi2025b}, two tests are conducted on the standard deviation: (i) an uncertain homogeneity test $\sigma_1 = \sigma_2 = \cdots = \sigma_r=\sigma$; (ii)  an uncertain common test  $\sigma=\sigma_0$ under (i) holds.
Only after these validations can uncertain confidence intervals and hypothesis tests be performed. The same procedure applies to the subsequent models and thus is not repeated hereafter.}
\end{remark}

Since $\varepsilon_{ij} \sim \mathscr{N}(0, \sigma_0)$, it is obvious that  $z_{ij} \sim \mathscr{N}(\mu_i,\sigma_0)$. By the operational law,
$\hat{\mu}_i$ has the uncertainty distribution $\mathscr{N}(\mu_i,\sigma_0)$, i.e.,
\[
\Phi(z; \mu_i, \sigma_0) = \left(1 + \exp\left(\frac{\pi(\mu_i - z)}{\sqrt{3}\sigma_0}\right)\right)^{-1}.
\]
Then,  following \citet{LioLiu2018}, the $\alpha$ (e.g., 95\%) confidence interval of $\mu _i$ is
\begin{equation}
\label{muici}
CI(\mu_i)=\left[\Phi^{-1}\left(\frac{1-\alpha}{2}; \hat{\mu}_i, \sigma_0\right),\Phi^{-1}\left(\frac{1+\alpha}{2}; \hat{\mu}_i, \sigma_0\right)\right]=\hat{\mu}_i \pm \sigma_0\frac{\sqrt{3}}{\pi}\ln_{}{\left(\frac{1+\alpha}{1-\alpha}\right)}.
\end{equation}

Let \(w_i\) be the weight of the total number of experiments contributed by the replications of level \(A_i\), given by $w_i = m_i/N(i = 1, 2, \dots, r)$. 
Denote 
\[
\mu = \sum_{i=1}^r w_i \mu_i, \quad \sum_{i=1}^r w_i = 1.\]
We call $a_i = \mu_i - \mu (i = 1, 2, \dots, r)$ the effect of level \(A_i\), which is the deviation of the mean of level \(A_i\) from the overall mean. Obviously, $\sum_{i=1}^r w_i a_i = 0.$ Thus, the uncertain model  (\ref{Mmodel}) is equivalent to the model
\begin{eqnarray}
\label{Fmodel}
\left\{\begin{array}{ll}
\displaystyle z_{ij} = \mu + a_i + \varepsilon_{ij},  i = 1, 2, \dots, r, \ j = 1, 2, \dots, m_i,\\
\displaystyle \sum_{i=1}^rw_ia_i=0,  \varepsilon_{ij}\sim\mathscr{N}(0,\sigma_0) \text{ are mutually independent}.
\end{array}
\right.
\end{eqnarray}
We refer to the model (\ref{Fmodel}) as an uncertain fixed-effects (UFE) model based on a single-factor design, denoted by the SUFE model. 

\begin{thm}\label{thm2}
In the SUFE model (\ref{Fmodel}), let  $\hat\mu$ and $\hat a_i(i=1,\ldots, r)$ be the least-squares estimators of $\mu$ and $a_i$. Then,
\begin{equation}
\label{mu}
\hat{\mu} = \frac{1}{N} \sum_{i=1}^r \sum_{j=1}^{m_i} z_{ij}, \quad \hat{a}_i= \left( \frac{1}{m_i} - \frac{1}{N} \right)\sum_{j=1}^{m_i}  z_{ij} -\frac{1}{N} \sum_{k \neq i} \sum_{j=1}^{m_k} z_{kj}.
\end{equation}
\end{thm}
\begin{proof}
By employing the least-squares method,  $\hat\mu$ and $\hat a_i(i=1,\ldots, r)$ are obtained by solving
\[
(\hat{\mu},\hat{\mathbf{a}}) =\mbox{argmin}_{\mu, \bf a} \sum_{i=1}^r \sum_{j=1}^{m_i} \left( z_{ij} - \mu- a_i \right)^2,
\]
where \(\hat{\mathbf{a}}=(\hat{a}_1,\hat{a}_2,...,\hat{a}_r)\). We can obtain 
\[
\hat{\mu} = \frac{1}{N} \sum_{i=1}^r \sum_{j=1}^{m_i} z_{ij}, \quad \hat{a}_i= \frac{1}{m_i} \sum_{j=1}^{m_i} z_{ij} - \frac{1}{N} \sum_{i=1}^{r} \sum_{j=1}^{m_i} z_{ij}.
\]
To facilitate the direct derivation of their uncertainty distributions and to provide a basis for constructing uncertain confidence intervals, we express each estimator as a sum of independent uncertain normal variables, which yields (\ref{mu}).
\end{proof}

By the operational law, it follows that
\(\hat{\mu}\sim \mathscr{N}(\mu, \sigma_0)\) and \(\hat{a}_i\sim \mathscr{N}(a_i, 2\left(1 - m_i/N\right) \sigma_0)\), i.e.,
\begin{align*}
\Phi(z; \mu, \sigma_0) &= \left( 1 + \exp\left( \frac{\pi(\mu - z)}{\sqrt{3}\sigma_0} \right) \right)^{-1},\\
\Phi\left(z; a_i, 2\left(1 - \frac{m_i}{N}\right)\sigma_0\right)& = \left( 1 + \exp\left( \frac{\pi(a_i - z)}{2\sqrt{3}\left(1 - \frac{m_i}{N}\right)\sigma_0} \right) \right)^{-1}.
\end{align*}
Then,  the \(\alpha\) confidence intervals of \(\mu\)  and  \(a_i\)  are
\begin{align}
CI(\mu)&=\left[ \Phi^{-1} \left( \frac{1-\alpha}{2}; \hat{\mu}, \sigma_0 \right), \Phi^{-1} \left( \frac{1+\alpha}{2}; \hat{\mu}, \sigma_0 \right) \right] = \hat{\mu} \pm \sigma_0 \frac{\sqrt{3}}{\pi} \ln \left( \frac{1 + \alpha}{1 - \alpha} \right),\label{muci}\\
CI(a_i) &= \left[ \Phi^{-1} \left( \frac{1-\alpha}{2}; \hat{a}_i, \sigma_{a_i} \right), \Phi^{-1} \left( \frac{1+\alpha}{2}; \hat{a}_i, \sigma_{a_i} \right) \right] = \hat{a}_i \pm \sigma_{a_i} \frac{\sqrt{3}}{\pi} \ln \left( \frac{1 + \alpha}{1 - \alpha} \right),\label{aici} 
\end{align}
where  \(\sigma_{a_i} = 2\left(1 - m_i/N\right)\sigma_0\).

\subsection{Uncertain hypotheses}
In this subsection, we investigate the equality of the means in the uncertain model of a single-factor design and equality of the effects in uncertain fixed-effects models. In this and subsequent sections concerning hypothesis testing, for simplicity, only unknown parameters are retained in the  inverse distribution according to two cases: (i) if the mean $e_0$ is known,  
$\Phi^{-1}(\alpha; e_0, \sigma_{j0})$ is replaced by 
$\Phi^{-1}(\alpha; \sigma_{j0})$, and (ii) if the standard deviation 
$\sigma_0$ is known,  $\Phi^{-1}(\alpha; a_{j0}, \sigma_0)$ is rewritten by $\Phi^{-1}(\alpha; a_{j0})$.

For the model (\ref{Mmodel}), the homogeneity hypotheses of means \(\mu_i\):
\begin{eqnarray*}
H_a: \mu_1 = \mu_2 = \cdots = \mu_r \text{ versus }
H_b: \mu_1, \mu_2, \ldots, \mu_r \text{ are not all equal}.
\end{eqnarray*}
If the null hypothesis $H_a$ is rejected at the significance level $\alpha$, the levels of factor $A$ differ significantly. Otherwise,  there is no significant difference.   Since $a_i = \mu_i - \mu$ $(i = 1, 2, \dots, r)$, the hypotheses  $H_a$ is equivalent to the following case 
\begin{equation}
\label{homo}
H_0: a_1 = a_2 = \cdots = a_r \quad \text{versus} \quad H_1: a_1,\ldots, a_r \text{ are not all equal}. 
\end{equation}
Under the constraint \(\sum_{i=1}^r w_i a_i = 0\), if the null hypothesis holds, i.e., \(a_1 = a_2 = \cdots = a_r \triangleq a\), then the constraint becomes \(a \sum_{i=1}^r w_i = a \cdot 1 = 0\), which implies \(a = 0\). Under the constraint, the hypothesis \(a_1 = a_2 = \cdots = a_r\)  is equivaltent to  \(a_1 = a_2 = \cdots = a_r = 0\). Consequently,  the homogeneity test of the effects \(a_i\) is sufficient to determine whether all effects are zero. The process can be extended to both balanced and unbalanced designs.

Let $z_{ij}(i = 1, \ldots, r; j = 1, \ldots, m_i)$  be observations satisfying $z_{ij} \sim \mathscr{N}(\mu_i, \sigma_i)$. Since independent samples of $\hat{a}_i$ are not directly observable, we rely on the original data $z_{ij}$ to construct the rejection region. 
A natural adjustment based on $a_i = \mu_i - \mu$ is 
$
\tilde{z}_{ij} = z_{ij} - \hat{\mu},
$
where  $\hat{\mu} \sim \mathscr{N}(\mu, \sigma_0)$.  
However, the estimator $\hat{\mu}$ is dependent on groups, because $\tilde{z}_{ij}$ and $\tilde{z}_{kl}$ ($k \neq i$) become correlated through $\hat{\mu}$.  

{To preserve independence, $z_{ij}$ should be adjusted only by a fixed constant.  Denote $a_i^* = a_i + (\mu - \mu_0)$. We therefore replace $\hat{\mu}$ with a fixed constant $\mu_0$ that can be estimated from the data but is treated as known in the testing procedure, yielding
\[
\tilde{z}_{ij} = z_{ij} - \mu_0 =a_i + \mu - \mu_0+ \varepsilon_{ij} =a_i^*+\varepsilon_{ij}.
\]
Then, $\tilde{z}_{ij}$ are mutually independent and follow $\mathscr{N}(a_i^*, \sigma_i)$. }

\begin{thm}\label{thm3}
For the hypothesis problem (\ref{homo}), the rejection region at significance level $\alpha$ is
\begin{equation*}
\label{rejection-ai}
W\!=\!\left\{
\begin{aligned}
&\!(\tilde{z}_{i1}, \tilde{z}_{i2}, \dots, \tilde{z}_{im_i})\!:\!\exists~i \!\neq\!j~(i,j\!\in\!\{1,2,...,r\})\text{ such that at least } \alpha \text{ of indexes}\!\\
& p \text{'s}\text{ with }1 \leq p \leq m_i \text{ satisfy }\tilde{z}_{ip }< \Phi^{-1} (\alpha/2; a_{j0})\text{ or } \tilde{z}_{ip} > \Phi^{-1} (1-\alpha/2; a_{j0})
\end{aligned}
\right\},
\end{equation*}
where $\Phi^{-1}(\alpha; a_{j0}) = a_{j0} + \sigma_{i0}\frac{\sqrt{3}}{\pi}\ln\left(\frac{\alpha}{1-\alpha}\right)$ and $a_{i0}$ and $\sigma_{i0}$ are estimators of $a_i^*$,  $\sigma_i$.
\end{thm}
\begin{proof}
Since $\mu_0$ is a fixed constant in the testing procedure, the adjustments 
$\tilde{z}_{ij} = a_i^*+\varepsilon_{ij},$
are independent and $\tilde{z}_{ij}\sim \mathscr{N}(a_i^*,\sigma_i)$.  Thus, the hypothesis  (\ref{homo})  is equivalent to $a_1^* = a_2^* = \cdots = a_r^*$.  Note that the normal uncertainty distribution family $\{\mathscr{N}(e, \sigma): e \in \mathbb{R}\}$ is nonembedded for all parameters at any significance level $\alpha$.  By Theorem 4 of \cite{ZhangLi2025b}, we obtain the 
rejection region $W$ at significance level $\alpha$.
\end{proof}

{

\begin{remark}
In the models (\ref{Mmodel}) and (\ref{Fmodel}), we use a single \(\sigma_0\) for estimating 
 various effects. It aims to simplify computation and aligns with the SUFE model's assumption of homogeneous standard deviations, thereby integrating information from all experimental units to yield more stable estimates. However, in the homogeneity test of means or effects, we use individual \(\sigma_i\) for each level to avoid the risk that a single \(\sigma_0\) might obscure true effects, thereby improving the precision of the tests and the reliability of the test results.
\end{remark} 

{
\begin{remark}
For all homogeneity tests, the rejection criterion of the null hypothesis \(H_0\) can be equivalently expressed via the acceptance interval (AI). Specifically, for a dataset \((z_{i1}, z_{i2}, \ldots, z_{im_i})\), the acceptance interval for a parameter \(\theta_{j0}\) is defined as:  
\[
AI(z_i; \theta_{j0}) = \left[ \Phi^{-1} \left( \alpha/2; \theta_{j0} \right), \Phi^{-1} \left( 1 - \alpha/2; \theta_{j0} \right) \right].
\]  
If any $z_i$ has more than $\alpha m_i$ data outside its corresponding $AI(z_i; \theta_{j0})$, $H_0$ is rejected.
\end{remark}
}

\section{Uncertain fixed-effects model with two factors}\label{sec4}
Let \(A\) and \(B\) be two factors of interest in an experiment, each having \(A_1, A_2, \dots, A_r\) and \(B_1, B_2, \dots, B_s\), where \(r\) and \(s\) are the number of levels. Each combination pairs one level of $A$ with one level of $B$, written as $A_i B_j$ ($i=1,2,\ldots,r, j=1,2,\ldots,s$). The experiment is repeated \(m_{ij}\) times for each combination $A_i B_j$, with a total of \(N = \sum_{i=1}^{r} \sum_{j=1}^{s} m_{ij}\). Let \(z_{ijl}\) denote observations of the \(l\)th\((l=1,2,\ldots, m_{ij})\) unit of \(A_iB_j\), and all \(z_{ijl}\) from distinct \(A_iB_j\) are independent. We consider two cases in the following subsections: (i) two-factor design without interaction, and (ii) two-factor design with interaction.
\subsection{Two-factor design without interaction}\label{sec42}
In the two-factor model, assume the two factors are independent and that their effects are additive, with no interactions.
Let  \(\mu_{i \cdot }(i=1,2,\ldots,r)\) and \(\mu_{ \cdot  j}(j=1,2,\ldots,s)\) be the mean  of \(A_i\) and \(B_j\), and \(\mu_{  }\) be the overall mean. Thus, the mean of a combination $A_i B_j$ is expressed as:  
   \[
   \mu_{ij} = \mu + a_{i} + b_{j}, i=1,2,\ldots,r, j=1,2,\ldots,s,
   \] 
where \(a_{i} = \mu_{i\cdot} - \mu\), and \(b_{j} = \mu_{\cdot j} - \mu\) are the main effects of factors \(A\) and  \(B\).
Denote \( m_{i\cdot} = \sum_{j=1}^{s} m_{ij}\) and \(m_{\cdot j} = \sum_{i=1}^{r} m_{ij}\). Let \( w_{i\cdot} \) and \( w_{\cdot j} \) be the proportion of total observations that belong to \(A_i\) and \(B_j\); that is, \( w_{i\cdot} = m_i./{N}, w_{\cdot j} = m_{\cdot j}/{N} \). In the case of a balanced design, \( w_{i\cdot} = 1/{r}, w_{\cdot j} = 1/{s} \). Then, a UFE model without interaction is defined as follows:
\begin{eqnarray}
\label{two-factor model without interaction}
\left\{\begin{array}{ll}
&\displaystyle z_{ijl} = \mu_{  } + a_{i } + b_{  j} + \varepsilon_{ijl}, i=1,2,\ldots,r, j=1,2,\ldots,s, l=1,2,\ldots,m_{ij},\\
&\displaystyle \sum_{i=1}^r w_{i\cdot}a_{i } = 0, \sum_{j=1}^s w_{\cdot j}b_{  j} = 0, \varepsilon_{ijl} \sim \mathscr{N}(0, \sigma_0) \text{ are mutually independent}.
\end{array}
\right.
\end{eqnarray}

\begin{thm}\label{thm5}{
In the model (\ref{two-factor model without interaction}) with a balanced design, i.e.,
$m_{ij}=m$$(i=1,\ldots, r,j=1,\ldots,s)$, let $\hat\mu$, $\hat a_i$, and $\hat b_j$be the least-squares  estimators of $\mu$, $a_i$, and $b_j$. Then,}
\begin{equation}
\label{multimu}
\hat{\mu}_{} = \frac{1}{rsm} \sum_{i=1}^r \sum_{j=1}^s  \sum_{l=1}^{m} z_{ijl},\quad
\hat{a}_i = \frac{r - 1}{rsm} \sum_{j=1}^{s} \sum_{l=1}^{m} z_{ijl} - \frac{1}{rsm} \sum_{k \neq i} \sum_{j=1}^{s} \sum_{l=1}^{m} z_{kjl},
\end{equation}
\begin{equation}
\label{bi}
\hat{b}_j = \frac{s - 1}{rsm} \sum_{i=1}^{r} \sum_{l=1}^{m} z_{ijl} - \frac{1}{rsm} \sum_{i=1}^{r} \sum_{k \ne j} \sum_{l=1}^{m} z_{ikl}.
\end{equation}
\end{thm}
\begin{proof}
By employing the least-squares method,  $\hat\mu$, $\hat a_i$, and $\hat b_j$$(i=1,\ldots, r,j=1,\ldots,s)$ are obtained by solving
\[
(\hat{\mu},\hat{\mathbf{a}},\hat{\mathbf{b}}) =\mbox{argmin}_{\mu, \bf a, \bf b} \sum_{i=1}^r \sum_{j=1}^s \sum_{l=1}^{m} \left( z_{ijl} - \mu_{} - a_{i} - b_{ j} \right)^2,
\]
where \(\hat{\mathbf{a}}=(\hat{a}_1,\hat{a}_2,\ldots,\hat{a}_r)\) and \(\hat{\mathbf{b}}=(\hat{b}_1,\hat{b}_2,\ldots,\hat{b}_s)\). Hence,
\begin{equation}\label{ai,bj}
\hat{a}_i = \frac{1}{sm} \sum_{j=1}^{s} \sum_{l=1}^{m} z_{ijl} - \frac{1}{rsm} \sum_{i=1}^{r} \sum_{j=1}^{s} \sum_{l=1}^{m} z_{ijl}, 
\hat{b}_j = \frac{1}{rm} \sum_{i=1}^{r} \sum_{l=1}^{m} z_{ijl} - \frac{1}{rsm} \sum_{i=1}^{r} \sum_{j=1}^{s} \sum_{l=1}^{m} z_{ijl}.
\end{equation}
Further, we write each estimator as a sum of independent normal uncertain variables, so that the operational law can be applied to derive their uncertainty distributions.
\end{proof}

It follows from  the operational law 
that \(\hat{\mu}\sim\mathscr{N}(\mu, \sigma_0)\), \(\hat{a}_i\sim\mathscr{N}(a_i, 2(1-1/r)\sigma_0)\), and \(\hat{b}_j\sim\mathscr{N}(b_j, 2(1-1/s)\sigma_0)\). Then, the \(\alpha\) confidence intervals of \(\mu,a_i\) and \(b_j\) are
\begin{equation}
\label{multimuci}
CI(\mu)=\hat{\mu}_{} \pm \sigma_0 \frac{\sqrt{3}}{\pi} \ln \left( \frac{1 + \alpha}{1 - \alpha} \right),
\end{equation}
\begin{equation}
\label{aibi}
CI(a_i)=\hat{a}_{i} \pm \sigma_{a_i} \frac{\sqrt{3}}{\pi} \ln \left( \frac{1 + \alpha}{1 - \alpha} \right),
\quad
CI(b_j)=\hat{b}_{ j} \pm \sigma_{b_j} \frac{\sqrt{3}}{\pi} \ln \left( \frac{1 + \alpha}{1 - \alpha} \right),
\end{equation}
where \(\sigma_{a_i}=2(1-1/r)\sigma_0\) and \(\sigma_{b_j}=2(1-1/s)\sigma_0\).

{ Theorem \ref{thm5} is only available to the balanced case, but not to the unbalanced one. Therefore, in the case of imbalance, the Lagrangian multiplier method is proposed to estimate the parameters in an unbalanced model (\ref{two-factor model without interaction}). Let \( \mathbf{Z} \) be an \( N \times 1 \) response vector, 
\( \mathbf{X} \) be an \( N \times p \) matrix  for $p =  r + s + 1 $, that is,
\begin{align}
\mathbf{Z} &= (z_{111}, \ldots, z_{11m_{11}}, \ldots, z_{1s1}, \ldots, z_{1sm_{1s}}, \dots ,z_{r11}, \ldots, z_{r1m_{r1}},  \ldots, z_{rs1}, \ldots, z_{rsm_{rs}})^T, \nonumber \\[0.5em]
\mathbf{X} &=(\boldsymbol{x}_{111}, \ldots, \boldsymbol{x}_{11m_{11}},  \ldots, \boldsymbol{x}_{1s1}, \ldots, \boldsymbol{x}_{1sm_{1s}}, \dots , \boldsymbol{x}_{r11}, \ldots, \boldsymbol{x}_{r1m_{r1}}, \ldots, \boldsymbol{x}_{rs1}, \ldots, \boldsymbol{x}_{rsm_{rs}})^T \nonumber
\end{align}
such that  $\boldsymbol{x}_{ijl}\boldsymbol{\beta} = \mu + a_i + b_j$, where $\boldsymbol{x}_{ijl} = (1 , 0, \dots, 1, \dots, 0 , 0, \dots, 1, \dots, 0)_{1\times p}$, corresponding to parameters $\boldsymbol{\beta} = (\mu, a_1, \ldots, a_r, b_1, \ldots, b_s)^T$. 
Denote $\boldsymbol{\varepsilon} = (\varepsilon_1, \ldots, \varepsilon_N)^T$. Thus, the model (\ref{two-factor model without interaction}) is equivalent to the following system
\begin{eqnarray}
\label{matrix without}
 \mathbf{Z}=\mathbf{X}\boldsymbol{\beta}+\boldsymbol{\varepsilon},\quad
 \mathbf{C}\boldsymbol{\beta} = \mathbf{d},
\end{eqnarray}
where 
\[
\mathbf{C} = 
\begin{bmatrix}
0 & w_{1\cdot} & \cdots & w_{2\cdot} & 0 &  \cdots& 0 \\
0 & 0 &  \cdots& 0 & w_{\cdot1} &  \cdots& w_{\cdot2}
\end{bmatrix}_{2\times p},\quad
\mathbf{d} = 
(0,0)^T.
\]}

\begin{thm}\label{upe}
In the model (\ref{matrix without}) with an unbalanced design, let $\boldsymbol{\hat{\beta}}$ be the least-squares estimators of the parameter vector $\boldsymbol{\beta}$. Then,  
\begin{equation}\label{hatbeta}
\boldsymbol{\hat{\beta}}=(\mathbf{X}^T\mathbf{X})^+\mathbf{X}^T\mathbf{Z} -\frac{1}{2} (\mathbf{X}^T\mathbf{X})^+\mathbf{C}^T\boldsymbol{\hat{\lambda}},
\end{equation}
where $\boldsymbol{\hat{\lambda}}$ is the vector of estimated Lagrange multipliers, and  $A^+$ denotes the generalized inverse of matrix $A$.
\end{thm}
\begin{proof}
By employing the Lagrangian multiplier method,  $\hat{\boldsymbol{\beta}}$ and $\hat{\boldsymbol{\lambda}}$
are obtained by solving
\[
(\hat{\boldsymbol{\beta}},\hat{\boldsymbol{\lambda}}) = \mbox{argmin}_{\boldsymbol{\beta}, \boldsymbol{\lambda}} [(\mathbf{Z} - \mathbf{X}\boldsymbol{\beta})^T (\mathbf{Z} - \mathbf{X}\boldsymbol{\beta}) + \boldsymbol{\lambda}^T (\mathbf{C}\boldsymbol{\beta} - \mathbf{d})].
\]
Through calculation, we have
$
-2\mathbf{X}^T(\mathbf{Z} - \mathbf{X}\hat{\boldsymbol{\beta}}) + \mathbf{C}^T\hat{\boldsymbol{\lambda}}=0 ,   \mathbf{C}\hat{\boldsymbol{\beta}} = \mathbf{d}, 
$ 
which yields the result.    
\end{proof}

To obtain confidence intervals of the parameters, we need their distributions of the least-squares estimators. 
{
\begin{thm}\label{thm2}
{\rm In the model (\ref{matrix without}) with an unbalanced design, the following results hold.

(i)  \(\boldsymbol{\varepsilon} \sim \mathscr{N}(0, \sigma_0\mathbf{I}_N)\), where \(\mathbf{I}_N\) is the \(N\times N\) identity matrix.

(ii)  \(\mathbf{Z} \sim \mathscr{N}(\mathbf{X}\boldsymbol{\beta}, \sigma_0\mathbf{I}_N)\).

(iii) $\hat{\mu} \sim \mathscr{N}(\mu, \!\sum_{j=1}^N |q_{1j}|\sigma_0)$, 
\(\hat{a}_i\!\sim\!\mathscr{N}(a_i, \!\sum_{j=1}^N |q_{i+1,j}|\sigma_0),
\hat{b}_j\sim\!\mathscr{N}(b_j, \!\sum_{k=1}^N |q_{r+j+1,k}|\sigma_0).
\)}
\end{thm}
\begin{proof}
(i) Let \(\boldsymbol{\varepsilon}=\boldsymbol{\sigma}\boldsymbol{\tau}\), where
\(\boldsymbol{\tau}\) is an \(N\times1\) standard normal uncertain vector, and 
\(\boldsymbol{\sigma} = \operatorname{diag}(\sigma_0, \sigma_0, \dots, \sigma_0)\)  is an \(N\times N\) diagonal matrix.
By Theorem 3 (\cite{Liu2013}) on joint uncertainty distributions,
\(\boldsymbol{\varepsilon}\) follows a multivariate normal distribution
$\Phi_{\boldsymbol{\varepsilon}}(x_1, x_2, \cdots, x_N) $ = $ \bigwedge_{i=1}^N\Phi_{\varepsilon_i}(x_i;\sigma_0),$
where
\[
\Phi_{\varepsilon_i}(x_i;\sigma_0)=\left( 1 + \exp\left( \frac{-\pi x_i}{\sqrt{3}\sigma_0} \right) \right)^{-1}.
\]
Thus, \(\boldsymbol{\varepsilon} \sim \mathscr{N}(0, \sigma_0\mathbf{I}_N)\).

(ii) Denote \(\boldsymbol{e}=\mathbf{X}\boldsymbol{\beta}=(e_1,e_2,\dots,e_N)^T\).  It is obvious that \(
\mathbf{Z}=\mathbf{X}\boldsymbol{\beta}+\boldsymbol{\varepsilon}=\boldsymbol{e}+\boldsymbol{\varepsilon}.\) 
By Theorem 6 (\cite{Liu2013}) on linear transformations of normal uncertain vectors, \(\mathbf{Z}\) is a normal uncertain vector. From Theorem 3 (\cite{Liu2013}) on joint uncertainty distributions, \(\mathbf{Z}\) follows a multivariate normal distribution
\(
\Phi_{\mathbf{Z}}(x_1, x_2, \cdots, x_N) = \bigwedge_{i=1}^N\Phi_{z_i}(x_i;e_i,\sigma_0),
\)
where
\[
\Phi_{z_i}(x_i;e_i,\sigma_0)=\left( 1 + \exp\left( \frac{\pi(e_i - x_i)}{\sqrt{3}\sigma_0} \right) \right)^{-1}.
\]
Hence, \(\mathbf{Z} \sim \mathscr{N}(\mathbf{X}\boldsymbol{\beta}, \sigma_0\mathbf{I}_N)\).

(iii) From Theorem \ref{upe}, we have
\(
\hat{\boldsymbol{\beta}}=(\mathbf{X}^T\mathbf{X})^+ \mathbf{X}^T \mathbf{Z} - (\mathbf{X}^T\mathbf{X})^+ \mathbf{C}^T \hat{\boldsymbol{\lambda}} = \mathbf{QZ} + \mathbf{P},
\)
where \(\mathbf{Q}=(\mathbf{X}^T\mathbf{X})^+ \mathbf{X}^T = (q_{ij})_{p\times N}\) and
\(\mathbf{P}=-(\mathbf{X}^T\mathbf{X})^+ \mathbf{C}^T \hat{\boldsymbol{\lambda}} = (P_1,P_2,\dots,P_p)^T\).
From (ii), \(\mathbf{Z} \sim \mathscr{N}(\mathbf{X}\boldsymbol{\beta}, \sigma_0\mathbf{I}_N)\), that is,
\(
\mathbf{Z} = \mathbf{X}\boldsymbol{\beta} + \sigma_0 \boldsymbol{\tau}
\). Substituting yields that
\[
\hat{\boldsymbol{\beta}} = \mathbf{Q}(\mathbf{X}\boldsymbol{\beta} + \sigma_0 \boldsymbol{\tau}) + \mathbf{P} = (\mathbf{QX}\boldsymbol{\beta} + \mathbf{P}) + \sigma_0 \mathbf{Q}\boldsymbol{\tau}.
\]
By the constructive property of Theorem \ref{upe}, when \(\mathbf{Z}=\mathbf{X}\boldsymbol{\beta}\) (i.e., \(\boldsymbol{\tau}=\mathbf{0}\)), we have \(\hat{\boldsymbol{\beta}}=\boldsymbol{\beta}\). Hence, \(\mathbf{QX}\boldsymbol{\beta} + \mathbf{P} = \boldsymbol{\beta}\), and
$
\hat{\boldsymbol{\beta}} = \boldsymbol{\beta} + \sigma_0 \mathbf{Q} \boldsymbol{\tau}.
$
By Theorem 6 (\cite{Liu2013}) on linear transformations of normal uncertain vectors, 
$\hat{\boldsymbol{\beta}}$ is still a normal uncertain vector.
For each component of \(\hat{\boldsymbol{\beta}}\),
$
\hat{\beta}_i = \beta_i + \sigma_0 \sum_{j=1}^N q_{ij} \tau_j.
$
By the operational law,
\begin{equation}
\label{betai}
\hat{\beta}_i \sim \mathscr{N}(\beta_i, \sum_{j=1}^N |q_{ij}|\sigma_0).
\end{equation}
Since \(\boldsymbol{\hat{\beta}}=(\hat{\beta}_1,\ldots,\hat{\beta}_p)^T=(\hat{\mu},\hat{a}_1,\ldots, \hat{a}_r, \hat{b}_1,\ldots, \hat{b}_s)^T\), the distributions for \(\hat{\mu}\), \(\hat{a}_i\) and \(\hat{b}_j\) follow immediately.
\end{proof}
}

Based on Theorem \ref{thm2}, the \(\alpha\) confidence intervals are obtained as follows
\begin{equation}
\label{multimuci-un}
CI(\mu)=\hat{\mu}_{} \pm \sigma_{\mu} \frac{\sqrt{3}}{\pi} \ln \left( \frac{1 + \alpha}{1 - \alpha} \right),
\end{equation}
\begin{equation}
\label{aibjci-un}
CI(a_i)=\hat{a}_{i} \pm \sigma_{a_i} \frac{\sqrt{3}}{\pi} \ln \left( \frac{1 + \alpha}{1 - \alpha} \right),
\quad
CI(b_j)=\hat{b}_{ j} \pm \sigma_{b_j} \frac{\sqrt{3}}{\pi} \ln \left( \frac{1 + \alpha}{1 - \alpha} \right),
\end{equation}
where \(\sigma_{\mu}=\sum_{j=1}^N |q_{1j}|\sigma_0\), \(\sigma_{a_i}=\sum_{j=1}^N |q_{i+1,j}|\sigma_0\) and \(\sigma_{b_j}=\sum_{k=1}^N |q_{r+j+1,k}|\sigma_0\). 

\begin{remark}
In the two-factor experiment, we employ a single \(\sigma_0\) for parameter estimation, as in the single-factor case. Still, in the testing phase, we use individual \(\sigma_{i\cdot}\) or \(\sigma_{\cdot j}\) for each factor's level, where \(\sigma_{i\cdot 0}\) and \(\sigma_{\cdot j 0}\) denote the standard deviations computed from the combined samples of  \(A_i\) and \(B_j\), respectively.
\end{remark} 

Next, we evaluate the equality of the effects under different levels. For the model (\ref{two-factor model without interaction}), the hypotheses are expressed as:
\begin{equation*}
\label{multi test for ai}
H_{0}^A: a_{1} = a_{2} = \cdots = a_{r}, \quad H_{1}^A: 
a_{1}, a_{2}, \dots, a_{r} \text{ are not all equal},
\end{equation*}
\begin{equation*}
\label{multi test for bj}
H_{0}^B: b_{1} = b_{2} = \cdots = b_{s}, \quad H_{1}^B: 
b_{1}, b_{2}, \dots, b_{s} \text{ are not all equal}.
\end{equation*}
Similar to the single-factor case, the homogeneity test of $a_i$ 
(resp. $b_j$) is sufficient to determine whether all effects are zero.
Before testing the homogeneity of main effects, observations from the same factor level need to be combined. Specifically, to isolate the effect of factor $A$ and perform the homogeneity test, observations at the same level of factor $A$ across all levels of factor $B$ are merged to form a complete sample set for testing, and similarly for factor $B$.
For instance, for the level $i$ of factor $A$:
\begin{equation}
\label{adjust-zij}
z_{ik} =  
z_{ij\left(k-\sum_{l=1}^{j-1}m_{il}\right)},~ \sum_{l=1}^{j-1}m_{il} < k \leq \sum_{l=1}^j m_{il}, j=1,\cdots,s, 
\end{equation}
and similarly for factor $B$ as $z_{jk}$. By collapsing the data by factor level, we reduce the homogeneity test for each main effect in a two-factor design to that of a single-factor effect. 
As described in Section \ref{sec3}, we replace $\hat{\mu}$ with a fixed constant $\mu_0$, yielding $\tilde{z}_{ik} = z_{ik} - \mu_0$ and $\check{z}_{jk} = z_{jk} - \mu_0$. The subsequent testing procedures are consistent with the single-factor case and will not be repeated here.
\subsection{Two-factor design with interaction}\label{sec41}
{In practical experiments, the influence of one factor may depend on the level of another factor, leading to an interaction effect. In this subsection, we will establish an uncertain two-factor model with interaction to capture this non-additive relationship between factors. We extend the model (\ref{two-factor model without interaction}) by introducing the interaction term $AB$.
The weights \( w_{ij} = m_{ij}/{N}\) represent the proportion of total observations for the combination \(A_iB_j\). For any balanced design, \( w_{ij} = 1/{rs}\).
The complete UFE model with interaction is defined:
\begin{align}\label{Tww}
\left\{
\begin{aligned}
&z_{ijl} = \mu_{} + a_{i} + b_{ j} + (ab)_{ij} + \varepsilon_{ijl}, i=1,2,...,r, j=1,2,...,s, l=1,2,...,m_{ij}, \\
&\sum_{i=1}^r w_{i\cdot}a_{i} = 0, \sum_{j=1}^s w_{\cdot j}b_{ j} = 0, \sum_{i=1}^r w_{ij}(ab)_{ij} = 0,\sum_{j=1}^s w_{ij}(ab)_{ij} = 0, \\
&\varepsilon_{ijl} \sim \mathscr{N}(0, \sigma_0) \text{ are mutually independent}.
\end{aligned}
\right.
\end{align}

\begin{thm}\label{thm8}
In the model (\ref{Tww}) with a balanced design, i.e.,
$m_{ij}=m(i=1,\ldots,r,j=1,\ldots,s)$, let $\hat\mu$, $\hat a_i,\hat b_j$ and \(\hat{(ab)}_{ij}\) be the least-squares  estimators of $\mu$, $a_i$, $b_j$ and \((ab)_{ij}\). Then, the estimators  $\hat\mu$, $\hat a_i,\hat b_j$ are the same as  those of (\ref{multimu})-(\ref{bi}) and
\begin{align*}
\label{abij}
\begin{split}
\hat{(ab)}_{ij} = &\frac{rs - r - s + 1}{rsm} \sum_{l=1}^{m} z_{ijl} + \frac{1 - s}{rsm} \sum_{k \neq i} \sum_{l=1}^{m} z_{kjl} + \frac{1 - r}{rsm} \sum_{n \neq j} \sum_{l=1}^{m} z_{inl} \\
&+ \frac{1}{rsm} \sum_{k \neq i} \sum_{n \neq j} \sum_{l=1}^{m} z_{knl}.
\end{split}
\end{align*}
\end{thm}}
\begin{proof} Denote \(\hat{\mathbf{a}}=(\hat{a}_1,\hat{a}_2,...,\hat{a}_r)\), \(\hat{\mathbf{b}}=(\hat{b}_1,\hat{b}_2,\ldots,\hat{b}_s)\), \(\hat{\mathbf{ab}} = (\hat{(ab)}_{ij} \,|\, i = 1, 2, \ldots, r, \, j = 1, 2, \ldots, s)\).
By employing the least-squares method,  $\hat\mu$, $\hat a_i,\hat b_j$ and \(\hat{(ab)}_{ij}\) $(i=1,\ldots, r,j=1,\ldots,s)$ are obtained by solving
\[
(\hat{\mu},\hat{\mathbf{a}},\hat{\mathbf{b}},\hat{\mathbf{ab}}) =\mbox{argmin}_{\mu, \bf a, \bf b, \bf ab} \sum_{i=1}^r \sum_{j=1}^s \sum_{l=1}^{m} \left( z_{ijl} - \mu_{} - a_{i} - b_{ j} - (ab)_{ij} \right)^2
\]
We obtain the estimators  $\hat\mu$, $\hat a_i,\hat b_j$  and 
\[
\hat{(ab)}_{ij} = \frac{1}{m} \sum_{l=1}^{m} z_{ijl} 
- \frac{1}{sm} \sum_{j=1}^{s} \sum_{l=1}^{m} z_{ijl} 
- \frac{1}{rm} \sum_{i=1}^{r} \sum_{l=1}^{m} z_{ijl} 
+ \frac{1}{rsm} \sum_{i=1}^{r} \sum_{j=1}^{s} \sum_{l=1}^{m} z_{ijl}.
\]
Each estimator is written as a sum of independent normal uncertain variables, so that the operational law can be applied to derive their uncertainty distributions.
\end{proof}

Based on Theorem \ref{thm8}, it follows from the operational law that
\(\hat{(ab)}_{ij}\sim\mathscr{N}((ab)_{ij}, 4(1-1/r-1/s+1/(rs))\sigma_0\). Further, the \(\alpha\) confidence interval are (\ref{multimuci}), (\ref{aibi}) and
\begin{equation}
CI((ab)_{ij})=\label{abijci}
\hat{(ab)}_{ij} \pm 4\left(1-\frac{1}{r}-\frac{1}{s}+\frac{1}{rs}\right)\frac{\sqrt{3}}{\pi}\sigma_0 \ln \left( \frac{1 + \alpha}{1 - \alpha} \right).
\end{equation}

The unbalanced two-factor model with interaction is handled similarly to the unbalanced case without interaction using the Lagrangian multiplier method. The response vector \(\mathbf{Z}\) and a \((N \times p)\) matrix \(\mathbf{X}\) for \(p = 1 + r + s + rs\) are defined analogously, but \(\boldsymbol{\beta}\) and \(\boldsymbol{x}_{ijl}\)  include the interaction terms:
\begin{align*}
\boldsymbol{\beta} =& (\mu, a_1, \dots, a_r, b_1, \dots, b_s, (ab)_{11}, \dots, (ab)_{rs})^T,\\
\boldsymbol{x}_{ijl} =&(1, 0, \dots, 1, \dots, 0, 0, \dots, 1, \dots, 0 , 0, \dots, 1, \dots, 0)_{1 \times p}
\end{align*}
such that \(\boldsymbol{x}_{ijl}\boldsymbol{\beta} = \mu + a_i + b_j + (ab)_{ij}\).
The constraint \(((2+r+s) \times p)\) matrix \(\mathbf{C}\)  and \(\mathbf{d} = \mathbf{0}_{(2+r+s) \times 1}\) are extended accordingly to enforce the two main-effect constraints and the additional \(r+s\) constraints on the interaction terms, where
\[
\mathbf{C} = \begin{bmatrix}
\mathbf{0}_{2 \times 1} & \mathbf{C}_a & \mathbf{C}_b & \mathbf{0}_{2 \times rs} \\
\mathbf{0}_{r \times 1} & \mathbf{0}_{r \times r} & \mathbf{0}_{r \times s} & \mathbf{C}_{ab}^r \\
\mathbf{0}_{s \times 1} & \mathbf{0}_{s \times r} & \mathbf{0}_{s \times s} & \mathbf{C}_{ab}^s
\end{bmatrix}, \quad \mathbf{C}_a = \begin{bmatrix} w_{1 \cdot} & \cdots & w_{r \cdot} \\ 0 & \cdots & 0 \end{bmatrix}, \quad \mathbf{C}_b = \begin{bmatrix} 0 & \cdots & 0 \\ w_{\cdot 1} & \cdots & w_{\cdot s} \end{bmatrix}, 
\]
and $\mathbf{C}_{ab}^r$ is an $r\times rs$ matrix with the $(i,\,(i-1)s+j)$th entry equal to $w_{ij}$ and zeros elsewhere, and $\mathbf{C}_{ab}^s$ is an $s\times rs$ matrix with the $(j,\,(i-1)s+j)$th entry equal to $w_{ij}$ and zeros elsewhere.

The least-squares estimator \(\hat{\boldsymbol{\beta}}\) is obtained in exactly the same form as in the no-interaction case (\ref{hatbeta}). By the same argument as in Theorem \ref{thm2}, each component satisfies (\ref{betai}). In particular,
\[
\hat{(ab)}_{ij}=\hat{\beta}_{r+s+(i-1)s+j+1} \sim \mathscr{N}\Bigl((ab)_{ij},\ \sum_{k=1}^N |q_{r+s+(i-1)s+j+1,k}|\sigma_0\Bigr).
\]
The \(\alpha\) confidence intervals for \(\mu\), \(a_i\) and \(b_j\) retain the forms (\ref{multimuci-un}) and (\ref{aibjci-un}), while for the interaction effect we have
\begin{equation}
\label{abijci-un}
CI((ab)_{ij})=
\hat{(ab)}_{ij} \pm \sigma_{(ab)_{ij}} \frac{\sqrt{3}}{\pi} \ln \left( \frac{1 + \alpha}{1 - \alpha} \right),
\end{equation}
where \(\sigma_{(ab)_{ij}}=\sum_{k=1}^N |q_{r+s+(i-1)s+j+1,k}|\sigma_0\).
Note that \(q_{ij}\) in the standard deviations of \(\hat{\mu}\), \(\hat{a}_i\), and \(\hat{b}_j\) differ between the no-interaction and interaction models because of the expanded structure of \(\mathbf{X}\) in \(\mathbf{Q}=(\mathbf{X}^T\mathbf{X})^+ \mathbf{X}^T\).

The homogeneity tests for the main effects \(a_i\) and \(b_j\) proceed exactly as in the model without interaction. For testing the significance of the interaction,
\begin{equation}
\label{multi test for abij}
 H_{0}^{AB}: (ab)_{ij} = 0 ~ \forall i, j \quad \text{vs.} \quad H_{1}^{AB}: \exists ~(i,j) \text{ s.t.} \,(ab)_{ij} \neq 0.   
\end{equation}
Similarly, the adjusted observations \(\breve{z}_{ijl}=z_{ijl}-\mu_0-a_{i0}-b_{j0}\) preserve independence across all combinations, where \(\mu_0\), \(a_{i0}\), \(b_{j0}\) are treated as known fixed constants.

\begin{thm}\label{rejection region for abij}
For hypothesis problem (\ref{multi test for abij}), the rejection region at significance level \(\alpha\) is
\begin{equation*}
\label{rejection-abij}
W\!=\!\!\!\!\bigcup_{\substack{i=1,...,r \\ j=1,...,s}}\!\!\!\!W_{ij}\!=\!\!\! \!\bigcup_{\substack{i=1,...,r \\ j=1,...,s}}\!\!\left\{
\begin{aligned}
&(\breve{z}_{ij1}, \breve{z}_{ij2},..., \breve{z}_{ijm_{ij}})\!:\!\text{there are at least } \alpha \text{ of indexes }p \text{'s} \text{ with} \\
&\!1\!\leq \!p\! \leq\!m_{ij} \text{ such that }\breve{z}_{ijp} \!<\! \Phi^{-1} (\alpha/2)\text{ or } \breve{z}_{ijp} \!>\! \Phi^{-1} (1-\alpha/2)\!
\end{aligned}
\right\},
\end{equation*}
where \(
\Phi^{-1} (\alpha)=\frac{\sigma_{ij0}\sqrt{3}}{\pi}\ln_{}{(\frac{\alpha}{1-\alpha})}\),  and $\sigma_{ij0}$ are estimators of $\sigma_{ij}$.
\end{thm}
\begin{proof}
Since $\mu_0$, $a_{i0}$, and $b_{j0}$ are  fixed constants 
in the testing procedure, the adjustment $\breve{z}_{ijl}$ preserves mutual independence across 
different combinations $(i,j)$. The composite null hypothesis $H_0^{AB}: (ab)_{ij} = 0, \forall i,j$ can be expressed as the intersection of individual hypotheses $H^{ij}_0: (ab)_{ij} = 0$, i.e., $H_0^{AB} = \bigcap_{i,j} H^{ij}_0$. Let $W_{ij}$ denote the rejection region of $H^{ij}_0$ and $W_{ij}^c$ be the corresponding acceptance region. Then, the acceptance region for $H_0^{AB}$ is $\bigcap_{i,j} W_{ij}^c$. By De Morgan's law,  the rejection region for $H_0^{AB}$ is
$W = (\cap_{i,j} W_{ij}^c)^c = \cup_{i,j} W_{ij}.$ This implies that $H_0^{AB}$ is rejected if and only if $(\breve{z}_{ij1}, \breve{z}_{ij2}, \ldots, \breve{z}_{ijm_{ij}})$ fall in $W_{ij}$ for at least one combination $(i,j)$. 
Since each $W_{ij}$ is an independent two-sided hypothesis test following \citet{YeLiu2022}, the result follows. 
\end{proof}

If $H_0^{AB}$ is rejected, the interaction effect of factors $A$ and $B$ is significant at level $\alpha$.
\section{Real examples}\label{sec5}
This section provides several examples to illustrate the detailed processes of parameter estimation and hypothesis testing for uncertain single-factor and two-factor models, including both balanced and unbalanced designs.

\begin{exam} (A single-factor design)\label{exam1}
{\rm This study investigates the effect of caffeine concentration on the transport of labeled adenine across the blood-brain barrier from \citet{McCall1982} (Table \ref{tab1}). Let $A$ be a caffeine concentration factor with three levels: \(A_1 = 0.1 \, \text{mM}\), \(A_2 = 0.5 \, \text{mM}\), and \(A_3 = 10 \, \text{mM}\). 
\begin{table}[h]
\centering
\caption{Data on the concentration of labeled adenine in the rat brains.}
\label{tab1}
\begin{tabular*}{\textwidth}{@{\extracolsep\fill}ccccccccccccc@{\extracolsep\fill}}
\toprule
\hspace{2em} \(i\) & \(j\) &\multicolumn{6}{c}{\(z_{ij}\)}\\ 
\midrule
\hspace{2em} 1 & 1-5   & 2.91 & 4.14 & 6.29 & 4.40 & 3.77 \\
\hspace{2em} 2 & 1-4   & 5.80 & 5.84 & 3.18 & 3.18 \\
\hspace{2em} 3 & 1-6   & 3.05 & 1.94 & 1.23 & 3.45 & 1.61 & 4.32 \\
\bottomrule
\end{tabular*}
\end{table}

In Table \ref{tab1},  \(z_{ij}\) are the concentration of labeled adenine in the brain for the \(j\)th observation under level \(A_i\), \(\mu_i\) is the mean of the response for level \(A_i\). To analyze the effects of caffeine concentration, an uncertain model is constructed as follows:
\begin{equation*}
z_{ij} = \mu_i + \varepsilon_{ij}, \varepsilon_{ij}\sim \mathscr{N}(0, \sigma),
\end{equation*}
or  $z_{ij} = \mu + a_i + \varepsilon_{ij}, \sum_{i=1}^3 a_i = 0$
with the overall mean \(\mu\)  and the fixed effect \(a_i\)  of level \(i\).
Based on Theorem \ref{thm1}, we obtain the estimators:  $\hat{\mu}_1=4.302$, $\hat{\mu}_2=4.5$, and $\hat{\mu}_3=2.6$.

\begin{table}[htbp]
\centering
\caption{The hypothesis test of residuals.}
\label{tab2}
\begin{tabular*}{\textwidth}{@{\extracolsep\fill}cccccccccccccc@{\extracolsep\fill}}
\toprule
\hspace{2.5em}\(\varepsilon_{i\cdot}\) & \(e_{i0}\) & \(\sigma_{i0}\) & \(AI(\varepsilon_{i\cdot};e_{i0}, \sigma_{i0})\) & Singular Point \\
\midrule
\hspace{2.5em}\(\varepsilon_{1\cdot}\) & 0 & 1.114 & [-2.251, 2.251] & 0 \\
\hspace{2.5em}\(\varepsilon_{2\cdot}\) & 0 & 1.320 & [-2.666, 2.666] & 0 \\
\hspace{2.5em}\(\varepsilon_{3\cdot}\) & 0 & 1.094 & [-2.209, 2.209] & 0 \\
\bottomrule
\end{tabular*}
\end{table}
We test whether the residuals $\hat{\varepsilon}_{ij} = z_{ij} - \hat{\mu}_i$ follow $\mathscr{N}(0, \sigma_i)$. Table \ref{tab2} shows that each residual $\varepsilon_{i}\sim\mathscr{N}(0, \sigma_{i0})\).  Then, consider the following hypotheses: 
\[
H_0:\sigma_1=\sigma_2=\sigma_{3}\text{ versus } \sigma_1, \sigma_2,\sigma_{3}~\text{are not all equal}.
\]
The rejection region of \(H_0\) at the significance level \(\alpha=0.05\) is
 \begin{equation*}
W = \left\{
\begin{aligned}
&(\varepsilon_{i1}, \varepsilon_{i2}, \dots, \varepsilon_{im_i}):\exists~i \neq j~(i,j\in\{1,2,3\}) \text{ such that at least } \alpha \text{ of indexes} \\
&p \text{'s}\text{ with } 1 \leq p \leq m_i \text{ satisfy }\varepsilon_{ip} < \Phi^{-1} (\alpha/2;\sigma_{j0})\text{ or } \varepsilon_{ip} > \Phi^{-1} (1-\alpha/2;\sigma_{j0})
\end{aligned}
\right\},
\end{equation*}
where
\(
\Phi^{-1} (\alpha; \sigma_{j0})=\frac{\sigma_{j0}\sqrt{3}}{\pi}\ln_{}{(\frac{\alpha}{1-\alpha})} . 
\)
Table \ref{tab2} reveals that all \(\varepsilon_{i\cdot}\) fall in \(AI(\varepsilon_{i\cdot};\sigma_{j0})\)\((i\neq j, i,j=1,2,3)\), \(H_{0}\) can not be rejected. That is to say, there is no significant difference in the standard deviations of \(\varepsilon_i(i=1,2,3)\) at the significance level 0.05. Thus, the residuals satisfy the assumption of homogeneity of standard deviations. Further, an uncertain common test can be performed to determine whether the standard deviations \(\sigma_1=\sigma_2=\sigma_3=\sigma\) are equal to a fixed constant \(\sigma_0\). For this, we consider the hypotheses 
\begin{equation*}
H_{0}^{*}: \sigma = \sigma_0 \text{ versus } H_{1}^{*}: \sigma \neq \sigma_0.
\end{equation*}
The 3 sets of residuals are then combined into one set, i.e.,
\[
\{\varepsilon_{1}, \varepsilon_{2},..., \varepsilon_{15}\}=\{\varepsilon_{11},..., \varepsilon_{15},\varepsilon_{21},..., \varepsilon_{24}, \varepsilon_{31},..., \varepsilon_{36}\},
\]
and
\(
\sigma_0^2 = \frac{1}{15} \sum_{i=1}^{15} (\varepsilon_i)^2=1.165^2.
\)
Given a significance level \(\alpha=0.05\),  
\(
\Phi^{-1}(\alpha/2;\sigma_0)=-2.353, \Phi^{-1}(1-\alpha/2;\sigma_0)=2.353.
\)
Since \(\alpha \times 15 =0.75\), we have
\begin{equation*}
W = \left\{
\begin{aligned}
& (\varepsilon_{1}, \varepsilon_{2}, \dots, \varepsilon_{15}) : \text{there are at least } 1\text{ of indexes } p \text{'s with } \\
&\quad~~ 1 \leq p \leq 15 \text{ such that } \varepsilon_p < -2.353 \text{ or } \varepsilon_p > 2.353
\end{aligned}
\right\}.
\end{equation*}
Note that all the \(\varepsilon_j\) fall in \([-2.353, 2.353]\). Thus, \((\varepsilon_{1}, \varepsilon_{2},\ldots, \varepsilon_{15}) \notin W\),  and \(H_{0}^{*}\) cannot be rejected. Then, the standard deviation \(\sigma\) is equal to a fixed constant \(\sigma_0=1.165\). 

Under the independence assumption, we can obtain the 95\% uncertain confidence interval of \(\mu_i(i=1,2,3)\):  $4.302\pm 2.353$, $4.5\pm 2.353$, and $2.6\pm 2.353$ from (\ref{muici}).
Next, we test the homogeneity of means:
\[H_a:\mu_1=\mu_2=\mu_3~\text{versus}~H_b:\mu_1, \mu_2, \mu_3~\text{are not all equal}.\]
The rejection region of \(H_{a}\) at the significance level \(\alpha=0.05\) is 
\begin{equation*}
W = \left\{
\begin{aligned}
&(z_{i1}, z_{i2}, \dots, z_{im_i}):\exists~i \neq j~(i,j\in\{1,2,3\})\text{ such that at least } \alpha \text{ of indexes}\\
& p \text{'s}\text{ with }1 \leq p \leq m_i \text{ satisfy }z_{ip} < \Phi^{-1} (\alpha/2; \mu_{j0})\text{ or } z_{ip} > \Phi^{-1} (1-\alpha/2; \mu_{j0})
\end{aligned}
\right\},
\end{equation*}
where
\(
\Phi^{-1} (\alpha; \mu_{j0})=\mu_{j0}+\frac{\sigma_{i0}\sqrt{3}}{\pi}\ln_{}{(\frac{\alpha}{1-\alpha})} . 
\)
 Table \ref{tab3} shows that  \(z_{3j}(j=2,3,5)\) are not in \(AI(z_{3\cdot};\mu_{i0})(i=1,2)\), \(z_{1j}(j=3)\) not in \(AI(z_{1\cdot};\mu_{30})\),  and \(z_{2j}(j=1,2)\) not in \(AI(z_{2\cdot};\mu_{30})\).  It means that there is a significant difference in the means of \(A_3\) and \(A_i(i=1,2)\) at the significance level 0.05.  Thus, \(H_{a}\) should be rejected.

From (\ref{mu})--(\ref{aici}), the 95\% uncertain confidence interval of \(\mu\) and \(a_i(i=1,2,3)\)  satisfy: $3.674\pm 2.353$, $0.628\pm 3.137 $, $ 0.826\pm 3.451$, and $-1.074\pm 2.824$. Based on the estimators, the uncertain homogeneity test of effects can be conducted 
 \[
H_0:a_1=a_2=a_3~\text{versus}~H_1:a_1, a_2, a_3~\text{are not all equal}.
 \]
 \begin{table*}[h]
\centering
 \caption{The uncertain homogeneity test.}
 \label{tab3}
 \begin{tabular*}{\textwidth}{@{\extracolsep\fill}cccc@{\extracolsep\fill}}
 \toprule
\hspace{2.5em}$z_{i\cdot}$ or $\tilde{z}_{i\cdot}$ &$i=1$ & $i=2$&$i=3$ \hspace{2.5em} \\ 
 \midrule
 \hspace{2.5em} $AI(z_{i\cdot};\mu_{10})$    & [2.051, 6.553] & [1.636, 6.968]  & [2.093, 6.511] \hspace{2.5em} \\
 \hspace{2.5em} $AI(z_{i\cdot};\mu_{20})$    & [2.249, 6.751] & [1.834, 7.166]  & [2.291, 6.709] \hspace{2.5em} \\
 \hspace{2.5em} $AI(z_{i\cdot};\mu_{30})$    & [0.349, 4.851] & [-0.066, 5.266] & [0.391, 4.809] \hspace{2.5em} \\
\hspace{2.5em}$AI(\tilde{z}_{i\cdot};a_{10})$ & [-1.622, 2.878] & [-2.032, 3.288] & [-1.582, 2.838]\hspace{2.5em} \\ 
\hspace{2.5em}$AI(\tilde{z}_{i\cdot};a_{20})$ & [-1.424, 3.076] & [-1.834, 3.486] & [-1.384, 3.036]\hspace{2.5em} \\ 
\hspace{2.5em}$AI(\tilde{z}_{i\cdot};a_{30})$ & [-3.324, 1.176] & [-3.734, 1.586] & [-3.284, 1.136]\hspace{2.5em}  \\ 
\bottomrule
\end{tabular*}
\end{table*}
Using Theorem \ref{thm3}, we obtain the results shown in Table \ref{tab3}, where $\tilde{z}_{3j}$ $(j = 2,3,5)$ are not in $AI(\tilde{z}_{3\cdot}; a_{i0})$ $(i = 1,2)$, $\tilde{z}_{1j}$ $(j = 3)$ is not in $AI(\tilde{z}_{1\cdot}; a_{30})$, and $\tilde{z}_{2j}$ $(j = 1,2)$ are not in $AI(\tilde{z}_{2\cdot}; a_{30})$.
 Thus, there is a significant difference in the effects of \(A_3\) and \(A_i (i = 1, 2)\) at the significance level 0.05. It shows that \(H_0\) should be rejected.

In conclusion, caffeine concentration significantly affects the transport of labeled adenine across the blood-brain barrier. Statistical analysis also reveals a significant difference between the third concentration level ($A_3$) and the other two ($A_1$ and $A_2$). The conclusion drawn from Example \ref{exam1} is consistent with that obtained under the probability framework (\citet{McCall1982}).

}
\end{exam}

\begin{exam}(A two-factor balanced design)\label{exam2}
{\rm
This study investigates the effects of target density and the fraction of slash chips on the actual density of particleboards, as reported by \citet{Boehner1975} (Table \ref{Data on the density of particleboards}). The two factors are target density \(A\) and the fraction of slash chips \(B\). Target density 
has two levels: \(A_1 = 42 \, \text{lb/ft}^3\) and \(A_2 = 48 \, \text{lb/ft}^3\). The fraction of slash chips also has two levels: \(B_1 = 0\%\) and \(B_2 = 25\%\). 
\begin{table*}[!h]
 \centering
 \caption{Data on the density of particleboards.}
 \label{Data on the density of particleboards}
 \begin{tabular*}{\textwidth}{@{\extracolsep\fill}cccccc@{\extracolsep\fill}}
 \toprule
\hspace{2.5em} \(i\) & \(j\) & \(l\)&\multicolumn{3}{c}{\(z_{ijl}\)}\hspace{2.5em} \\ \midrule
\hspace{2.5em} \multirow{2}{*}{1} 
 & 1 & 1-3 & 40.9 & 42.8 & 45.4\hspace{2.5em} \\
 & 2 & 1-3 & 41.9 & 43.9 & 46.0\hspace{2.5em} \\
\hspace{2.5em} \multirow{2}{*}{2} 
 & 1 & 1-3 & 44.4 & 48.2 & 49.9\hspace{2.5em} \\
 & 2 & 1-3 & 46.2 & 48.6 & 50.8\hspace{2.5em} \\
 \bottomrule
 \end{tabular*}
\end{table*}

In Table \ref{Data on the density of particleboards}, \(z_{ijl}\) represents the actual density of the \(l\)th board under \(A_iB_j\). To analyze these factors, a balanced two-factor model with interaction is constructed as follows
\[
z_{ijl} = \mu + a_{i} + b_{ j} + (ab)_{ij} + \varepsilon_{ijl},\varepsilon_{ijl}\sim \mathscr{N}(0, \sigma),
\]
and \(\sum_{i=1}^2 w_{i\cdot}a_{i} = 0, \sum_{j=1}^2 w_{\cdot j}b_{ j} = 0, \sum_{i=1}^2 w_{ij}(ab)_{ij} = 0,\sum_{j=1}^2 w_{ij}(ab)_{ij} = 0\), where \(\mu\) is the overall mean, \(a_{i}\) is the main effect of \(A_i\), \(b_{ j}\) is the main effect of \(B_j\), and \((ab)_{ij}\) is the interaction effect between these two factor levels. 

Using Theorem \ref{thm8}, we obtain the  estimators: \(\hat{\mu}= 45.75\), \(\hat{a}_1= -2.267\), \(\hat{a}_2= 2.267\), \(\hat{b}_1= -0.483\), \(\hat{b}_2= 0.483\), \(\hat{(ab)}_{11}=0.033\), \(\hat{(ab)}_{12}=-0.033 \), \(\hat{(ab)}_{21}= -0.033 \), \(\hat{(ab)}_{22}= 0.033 \). We then test whether the residuals \(\hat{\varepsilon}_{ijl}=z_{ijl} -\hat{\mu} - \hat{a}_{i} - \hat{b}_{ j} - \hat{(ab)}_{ij}\) follow \(\mathscr{N}(0, \sigma_{ij})\).
Table \ref{Results of the hypothesis test for residuals in Example 2} shows that each residual $\varepsilon_{ij}\sim\mathscr{N}(0, \sigma_{ij0})$. Next, we consider the following hypotheses:
\[
H_0:\sigma_{11}=\sigma_{12}=\sigma_{21}=\sigma_{22}\text{ versus } H_1:\sigma_{11},\sigma_{12},\sigma_{21},\sigma_{22}~\text{are not all equal}.
\]
The rejection region of \(H_0\) at the significance level \(\alpha=0.05\) is similar to Example \ref{exam1}. 
Table \ref{Results of the hypothesis test for residuals in Example 2} reveals that all \(\varepsilon_{ij\cdot}\) fall in \(AI(\varepsilon_{ij\cdot};\sigma_{uv0})(ij \neq uv,ij,uv\in\{11,12,21,22\})\),  and \(H_{0}\) can not be rejected.  The residuals satisfy the assumption of homogeneity of standard deviations.

\begin{table}[!h]
\centering
\caption{Results of the hypothesis test for residuals.}
\label{Results of the hypothesis test for residuals in Example 2}
\begin{tabular*}{\textwidth}{@{\extracolsep\fill}cccccccccccccc@{\extracolsep\fill}}
\toprule
\hspace{2.5em}\(\varepsilon_{ij\cdot}\) & \(e_{ij0}\) & \(\sigma_{ij0}\) & \(AI(\varepsilon_{ij\cdot};e_{ij0}, \sigma_{ij0})\) & Singular Point \\ \midrule
\hspace{2.5em}\(\varepsilon_{11\cdot}\) & 0 & 1.845 & [-3.725, 3.725] & 0 \\ 
\hspace{2.5em}\(\varepsilon_{12\cdot}\) & 0 & 1.674 & [-3.381, 3.381] & 0 \\ 
\hspace{2.5em}\(\varepsilon_{21\cdot}\) & 0 & 2.299 & [-4.644, 4.644] & 0 \\ 
\hspace{2.5em}\(\varepsilon_{22\cdot}\) & 0 & 1.879 & [-3.794, 3.794] & 0 \\ 
\bottomrule
\end{tabular*}
\end{table}

Moreover, the uncertain common test is conducted to determine whether the standard deviations \(\sigma_{11}=\sigma_{12}=\sigma_{21}=\sigma_{22}=\sigma\) are equal to a fixed constant \(\sigma_0\). Based on this, we consider the following  hypotheses: 
\begin{equation*}
H_{0}^{*}: \sigma = \sigma_0 \text{ versus } H_{1}^{*}: \sigma \neq \sigma_0.
\end{equation*}
It is conducted in the same manner as in Example \ref{exam1}. 
All the \(\varepsilon_j\) fall in \([-3.914, 3.914]\). Thus, \((\varepsilon_{1}, \varepsilon_{2},\ldots, \varepsilon_{12}) \notin W\),  and \(H_{0}^{*}\) cannot be rejected. It means that the standard deviation \(\sigma_0=1.938\). Through (\ref{multimuci})-(\ref{aibi}) and (\ref{abijci}), the 95\% uncertain confidence interval of \(\mu\), \(a_i\), \(b_j\), and \((ab)_{ij}\) are: \(\mu= 45.75 \pm 3.915\), \(a_1=-2.267 \pm 3.915\), \(a_2= 2.267 \pm 3.915\), \(b_1=-0.483 \pm 3.915\), \(b_2= 0.483 \pm 3.915\), \((ab)_{11}=0.033 \pm 3.915\), \((ab)_{12}= -0.033 \pm 3.915\), \((ab)_{21}= -0.033 \pm 3.915\), \((ab)_{22}=0.033 \pm 3.915\).

After that, we proceed with the homogeneity test of the main effects and the significance test of the interaction effect.
The data corresponding to factor \( A \) are combined as (\ref{adjust-zij}), and similarly for factor \( B \).
\(z_{ik}\) and \(z_{jk}\) are combined data, where \( i = 1, 2 \) represents the levels of factor \( A \), \( j = 1, 2 \) represents the levels of factor \( B \), and \( k \) denotes the index ranging from 1 to \( m_i \) for \( i \) and from 1 to \( m_j \) for \( j \). This combination ensures that all relevant data is organized for subsequent tests.
We then consider the hypotheses (\ref{multi test for abij}) and
\[H_0^A:a_{1}=a_{2}~\text{versus}~H_1^A:a_{1} \neq a_{2},\quad H_0^B:b_{ 1}=b_{ 2}~\text{versus}~H_1^B:b_{ 1} \neq b_{ 2}.\]

After collapsing the observations by factor level, the rejection regions for testing $H_0^A$ and $H_0^B$ are constructed according to Theorem \ref{thm3}. 
In Table \ref{Results of the homogeneity tests for main effects in Example 2}, 
$\tilde{z}_{1k}(k=1,2,4,5)$ not in $AI(\tilde{z}_{1\cdot}; a_{20})$ and 
$\tilde{z}_{2k}(k=2,3,5,6)$ not in $AI(\tilde{z}_{2\cdot}; a_{10})$.
Thus, \(H_0^A\) is rejected, indicating that factor \(A\) is significant at the 0.05 level. 
In contrast, all $\check{z}_{j\cdot}$ fall in $AI(\check{z}_{j\cdot}; b_{i0})$, 
so $H_0^B$ cannot be rejected, indicating that factor $B$ is not significant at the 0.05 level.

Using Theorem \ref{rejection region for abij}, we obtain that \(\breve{z}_{11l}\) fall in [-3.727, 3.727], \(\breve{z}_{12l}\) fall in [-3.381, 3.381], \(\breve{z}_{21l}\) fall in [-4.644, 4.644], and \(\breve{z}_{22l}\) fall in [-3.795, 3.795]. It means that \(H_{0}^{AB}\) can not be rejected, that is, the interaction effect between the factors $A$ and $B$ are not significant.

\begin{table*}[!h]
\centering
\caption{Results of the homogeneity tests for main effects.}
\label{Results of the homogeneity tests for main effects in Example 2}
\begin{tabular*}{\textwidth}{@{\extracolsep\fill}ccc@{\extracolsep\fill}}
\toprule
\hspace{2.5em} $\tilde{z}_{i\cdot}$ or $\check{z}_{j\cdot}$ & $i$ or $j=1$ & $i$ or $j=2$ \hspace{2.5em} \\ \midrule
\hspace{2.5em} $AI(\tilde{z}_{i\cdot};a_{10})$ & [-5.939, 1.405] & [-6.634, 2.100] \hspace{2.5em} \\
\hspace{2.5em} $AI(\tilde{z}_{i\cdot};a_{20})$ & [-1.405, 5.939] & [-2.100, 6.634] \hspace{2.5em} \\ 
\hspace{2.5em} $AI(\check{z}_{j\cdot};b_{10})$ & [-6.976, 6.010] & [-6.724, 5.757] \hspace{2.5em} \\
\hspace{2.5em} $AI(\check{z}_{j\cdot};b_{20})$ & [-6.010, 6.976] & [-5.757, 6.724] \hspace{2.5em} \\ \bottomrule
\end{tabular*}
\end{table*}

In conclusion, target density has a significant impact on particleboard actual density, whereas the proportion of slash chip (0\% and 25\%) has no significant effect. No interaction was found between these factors. The results of  Example \ref{exam2} are consistent with those obtained under the probability framework (\citet{Boehner1975}). 
}
\end{exam}

\begin{exam}(A two-factor unbalanced design)\label{exam3}
{\rm
Consider the same two-factor interaction model as Example \ref{exam2}, but with an unbalanced design. The data is from \citet{Timm1998} (Table \ref{Unbalanced data in Example 3}). Each factor ($A$ and $B$) has two levels. Since the analysis follows the same procedures as in Example \ref{exam2}, the detailed steps are omitted here. The results are presented in Tables \ref{Results of the hypothesis test in Example 3}–\ref{Results of the homogeneity tests for main effects in Example 3}. 
At the significance level 0.05, the normality and homogeneity of the standard deviations assumptions are satisfied for residuals. Based on the uncertain common test, we have \(\sigma_0 = 7.429\). From \eqref{multimuci-un}--\eqref{aibjci-un} and \eqref{abijci-un}, the 95\% uncertain confidence interval of \(\mu\), \(a_i\), \(b_j\), and \((ab)_{ij}\) are: \(\mu=56.818 \pm 6.669\), \(a_1=11.852 \pm 10.003\), \(a_2= -14.222 \pm 10.003\), \(b_1=-2.040 \pm 10.003\), \(b_2=1.700 \pm 10.003\), \((ab)_{11}= -4.630 \pm 15.005\), \((ab)_{12}= 4.630 \pm 15.005\), \((ab)_{21}=6.944 \pm 15.005\), \((ab)_{22}=-4.630 \pm 15.005\). 
The homogeneity tests indicate that the main effects of factor $A$ are significant, whereas those of factor $B$ are not. Moreover, the interaction effect is significant, as \(\breve{z}_{212}\) lies outside the interval \([-11.109, 11.109]\).

\begin{table*}[!h]
 \centering
 \caption{Unbalanced data.}
 \label{Unbalanced data in Example 3}
 \begin{tabular*}{\textwidth}{@{\extracolsep\fill}cccccc@{\extracolsep\fill}}
 \toprule
\hspace{2.5em} \(i\) & \(j\) & \(l\)&\multicolumn{3}{c}{\(z_{ijl}\)}\hspace{2.5em} \\ \midrule
\hspace{2.5em} \multirow{2}{*}{1} 
 & 1 & 1-3 & 61 & 73 & 52\hspace{2.5em} \\
 & 2 & 1-3 & 79 & 65 & 81\hspace{2.5em} \\
\hspace{2.5em} \multirow{2}{*}{2} 
 & 1 & 1-2 & 42 & 53 &   \hspace{2.5em} \\
 & 2 & 1-3 & 37 & 32 & 50\hspace{2.5em} \\
 \bottomrule
 \end{tabular*}
\end{table*}

\begin{table}[!h]
\centering
\caption{Results of the hypothesis test for residuals.}
\label{Results of the hypothesis test in Example 3}
\begin{tabular*}{\textwidth}{@{\extracolsep\fill}cccccccccccccc@{\extracolsep\fill}}
\toprule
\hspace{2.5em}\(\varepsilon_{ij\cdot}\) & \(e_{ij0}\) & \(\sigma_{ij0}\) & \(AI(\varepsilon_{ij\cdot};e_{ij0}, \sigma_{ij0})\) & Singular Point \\ \midrule
\hspace{2.5em}\(\varepsilon_{11\cdot}\) & 0 & 8.602 & [-17.375, 17.375] & 0 \\ 
\hspace{2.5em}\(\varepsilon_{12\cdot}\) & 0 & 7.118 & [-14.377, 14.377] & 0 \\ 
\hspace{2.5em}\(\varepsilon_{21\cdot}\) & 0 & 5.500 & [-11.109, 11.109] & 0 \\ 
\hspace{2.5em}\(\varepsilon_{22\cdot}\) & 0 & 7.587 & [-15.323, 15.323] & 0 \\ 
\bottomrule
\end{tabular*}
\end{table}

\begin{table*}[!h]
\centering
\caption{Results of the homogeneity tests for main effects.}
\label{Results of the homogeneity tests for main effects in Example 3}
\begin{tabular*}{\textwidth}{@{\extracolsep\fill}ccc@{\extracolsep\fill}}
\toprule
\hspace{2.5em} $\tilde{z}_{i\cdot}$ or $\check{z}_{j\cdot}$ & $i$ or $j=1$ & $i$ or $j=2$ \hspace{2.5em} \\ \midrule
\hspace{2.5em} $AI(\tilde{z}_{i\cdot};a_{10})$ & [-8.804, 32.508] & [-3.969, 27.673] \hspace{2.5em} \\
\hspace{2.5em} $AI(\tilde{z}_{i\cdot};a_{20})$ & [-34.878, 6.434] & [-30.043, 1.599] \hspace{2.5em} \\
\hspace{2.5em} $AI(\check{z}_{j\cdot};b_{10})$ & [-22.929, 18.849] & [-40.693, 36.613] \hspace{2.5em} \\
\hspace{2.5em} $AI(\check{z}_{j\cdot};b_{20})$ & [-19.189, 22.589] & [-36.953, 40.353] \hspace{2.5em} \\ \bottomrule
\end{tabular*}
\end{table*}

Through calculating the expected response values for each combination,
we determine that, if the experimental objective is ``larger-the-better'', the optimal combination is $A_1B_2$, with an expected response value of 75.000; if the experimental objective is ``smaller-the-better'', the optimal combination is $A_2B_2$, with an expected response value of 39.666. Although the main effect of factor $B$ is not significant, its interaction with factor $A$ is significant, so the choice of factor $B$'s level still has an important impact on the results across different levels of factor $A$. 
Unlike the analysis under the probability framework (\citet{Timm1998}), which failed to detect a significant interaction, the uncertain model identified a significant interaction effect in this unbalanced design. This demonstrates the potential of the uncertain approach to reveal underlying effects that might be obscured in traditional probabilistic analysis, especially with small samples.
}
\end{exam}
\section{Conclusion}\label{sec6}
Within the framework of uncertainty theory, this paper conducted an in-depth analysis of uncertain fixed-effects models using single- and two-factor designs. Under UFE models,  the study provides a comprehensive solution for handling uncertain experimental data. These methodologies not only effectively address the imbalance in single-factor experimental designs but also address interaction effect and unequal replication in two-factor designs. Regarding the parameter estimation and testing problem,  several critical technical challenges have been overcome. For two-factor unbalanced designs, we ingeniously transformed the model into a multivariate linear regression model combined with uncertain vectors to obtain the parameter distributions, successfully resolving the limitations of traditional methods that could not derive explicit solutions. In the testing phase, we designed a novel data adjustment technique. By utilizing fixed constants, which can be estimated from data but are treated as known, to modify the data, we effectively maintained inter-group independence. More importantly, this study addresses a fundamental challenge in experimental design: the conservatism of probabilistic methods in small sample settings. The proposed UFE models, rooted in uncertainty theory, overcome limitations such as low power by not relying on large-sample asymptotics. Their thick-tailed distributions and test logic based on all data points yield a stricter, more robust inference that reliably captures true differences often obscured in traditional analyses. By analyzing three designs, we validated the effectiveness and practicality of the proposed methods. These cases encompassed single-factor unbalanced designs, two-factor balanced designs, and two-factor unbalanced designs, comprehensively demonstrating the adaptability and flexibility of the proposed methodologies across various scenarios.

Our proposed model and methods can be extended to other types of designs,  such as randomized block, split-plot, and orthogonal designs. Future investigations could focus on applying these methodologies.

\section*{Conflict of interest}
The authors declare that they have no conflict of interest.

\section*{Funding}
The work was supported by the 2025 Central Guidance for Local Science and Technology Development Fund (ZYYD2025ZY20), the National Natural Science Foundation of China (12561047),  and the Xinjiang Talent Development Fund (XJRC-2025-KJ-PY-KJLJ-108).

\end{document}